\documentclass[leqno,12pt]{amsart}
\usepackage{amscd} 
\usepackage{graphicx} 
\usepackage{tikz}
\usepackage{tikz-cd} 
\usepackage[margin=34mm]{geometry}
\usepackage{blaomnote}
\usepackage{hyperref}
\begin{document}
\bibliographystyle{plain}
\title%
{A geometer's view of the the Cram\'er-Rao bound on estimator
  variance} \author{Anthony D.~Blaom}
\address{E-mail: {\tt anthony.blaom@gmail.com}}%
\thispagestyle{empty}
\begin{abstract} 
  The classical Cram\'er-Rao inequality gives a lower bound for the
  variance of a unbiased estimator of an unknown parameter, in some
  statistical model of a random process. In this note we rewrite the
  statment and proof of the bound using contemporary geometric
  language.
\end{abstract}
\maketitle

The Cram\'er-Rao inequality gives a lower bound for the variance of a
unbiased estimator of a parameter in some statistical model of a
random process. Below is a restatement and proof in sympathy with the
underlying geometry the problem. While our presentation is mildly
novel, its mathematical content is very well-known. 

Assuming some very basic familiarity with Riemannian geometry, and
that one has reformulated the bound appropriately, the essential parts
of the proof boil down to half a dozen lines. For completeness we
explain the connection with log-likelihoods, and show how to recover
the more usual statement in terms of the Fisher information matrix.
We thank Jakob Str\"ohl for helpful feedback.

\section{The Cram\'er-Rao inequality}\label{_1}
The mathematical setting of statistical inference consists of: (i) a
smooth\footnote{In this note `smooth' means $C^2$.} manifold
${\mathcal X}$, the {\df sample space}, which we will suppose is
finite-dimensional; and (ii) a set ${\mathcal P}$ of probability
measures on ${\mathcal X} $, called the {\df space of models} or {\df
  parameters}. The objective is to make inferences about an unknown
model $p \in {\mathcal P} $, given one or more observations
$x \in {\mathcal X} $, drawn at random from ${\mathcal X} $ according
to $p$.

Under certain regularity assumptions detailed below, this data
suffices to make ${\mathcal X} $ into a Riemannian manifold, whose
geometric properties are related to problems of statistical
inference. It seems that Calyampudi Radhakrishna Rao was the first to
articulate this connection between geometry and statistics \cite{Rao_45}.

In formulating the Cram\'er-Rao inequality, we suppose that
${\mathcal P} $ is a smooth {\em finite}-dimensional manifold (i.e.,
we are doing so-called parametric inference). We say that
${\mathcal P} $ is {\df regular} if the probability measures
$p \in {\mathcal P} $ are all Borel measures on ${\mathcal X} $, and
if there exists some positive Borel measure $\mu$ on ${\mathcal X} $,
hereafter called a {\df reference measure}, such that
\begin{equation}
  p = f_p\,\mu, \label{hh}
\end{equation} 
for some collection of smooth functions $f_p$, $p \in {\mathcal P} $,
on ${\mathcal X} $. The definition of regularity furthermore requires
that we may arrange  $(x,p)\mapsto f_p(x)$ to be jointly smooth.

An {\df unbiased estimator} of some smooth function
$\theta \colon {\mathcal P} \rightarrow {\mathbb R} $ (the
``parameter'') is a smooth function
$\hat \theta \colon {\mathcal X} \rightarrow {\mathbb R} $ whose
expectation under each $p \in {\mathcal P} $ is precisely $\theta(p)$:
\begin{equation}
  \theta(p) = {\mathbb E} (\hat \theta\,|\,p) := \int_{x \in {\mathcal X}} \hat \theta
  (x)\,dp\,; \qquad p \in {\mathcal P}.\label{kk}
\end{equation}
\noindent%

\begin{theorem}[Rao-Cram\'er \cite{Rao_45,Cramer_46}]
  The space of models ${\mathcal P} $ determines a natural Riemannian
  metric on ${\mathcal X} $, known as the Fisher-Rao metric, with
  respect to which there is the following lower bound on the variance
  of an unbiased estimator $\hat \theta $ of $\theta $:
  \begin{equation}
     \label{eq:1}
     {\mathbb V}(\hat \theta\,|\,p) \ge |\nabla \theta (p)|^2\,; \qquad p \in {\mathcal P}.
   \end{equation}
\end{theorem}

\noindent%
More informally: The parameter space ${\mathcal P}$ comes equipped
with a natural way of measuring distances, leading to a well-defined
notion of steepest rate of ascent, for any function $\theta $ on
${\mathcal P} $. The square of this rate is precisely the lower bound
for the variance of an unbiased estimator $\hat\theta $.

\section{Observation-dependent one-forms on the space of
  models}\label{s2}
It is fundamental to the present geometric point of view that each
observation $x \in {\mathcal X} $ determines a one-form $\lambda_x$ on
the space ${\mathcal P} $ of models in the following way: Let $v \in T_{p_0}{\mathcal P} $
be a tangent vector, understood as the derivative of some path
$t \mapsto p_t \in {\mathcal P} $ through $p_0$:
\begin{equation}
  v = \frac{d}{dt} p_t \Big|_{t=0}.\label{gg}
\end{equation}
Then, recalling that each $p_t$ is a probability measure on
${\mathcal X} $ (and ${\mathcal P} $ is regular) we may write
$p_t=g_tp_0$, for some smooth function
$g_t \colon {\mathcal X} \rightarrow {\mathbb R} $, and define
\begin{equation*}
  \lambda_x(v) = \frac{d}{dt} g_t(x)\Big|_{t=0}.
\end{equation*}
The proof of the following is straightforward:
\begin{lemma}
  ${\mathbb E}(\lambda_x(v)\,|\,p) = 0$ for all $p \in {\mathcal P}$ and
  $v \in T_p {\mathcal P}$.
\end{lemma}

Now if $v \in T_{p_0} {\mathcal P} $ is a tangent vector as in
\eqref{gg}, and if \eqref{kk} holds, then
\begin{equation*}
  d \theta (v)=\frac{d}{dt}\int_{x \in {\mathcal X}}\hat \theta
  (x)\,d{p_t}\,\Big|_{t=0} =\frac{d}{dt}\int_{x \in {\mathcal X}}\hat \theta
  (x)g_t(x)\,d{p_0}\,\Big|_{t=0} = \int_{x \in {\mathcal X}}\hat
  \theta(x) \lambda_x(v)\,d{p_0},
\end{equation*}
giving us:
\begin{proposition}
  For any unbiased estimator
  $\hat \theta \colon {\mathcal X} \rightarrow {\mathbb R} $ of
  $\theta \colon {\mathcal P} \rightarrow {\mathbb R} $, one has
  \begin{equation*}
    d \theta (v) = \int_{x \in {\mathcal X}}\hat
    \theta(x) \lambda_x(v)\,d{p};\qquad v \in T_p {\mathcal P}.
  \end{equation*}
\end{proposition}

\section{Log-likelihoods}\label{s3}
As an aside, we shall now see that the observation-dependent one-forms
$\lambda_x$ are exact, and at the same time give their more usual
interpretation in terms of log-likelihoods.

Choosing a reference measure $\mu $, and defining $f_p$ as in
\eqref{hh}, one defines the {\df log-likelihood} function
$(x, p) \mapsto L_x(p) \colon {\mathcal X} \times {\mathcal P}
\rightarrow {\mathbb R} $ by
\begin{equation*}
  L_x(p) = \log f_p(x).
\end{equation*}
While the log-likelihood depends on the reference measure $\mu$, its
derivative $dL_x$ (a one-form on ${\mathcal P} $) does not, for in
fact:
\begin{lemma}
  $ dL_x = \lambda_x$.
\end{lemma}
\begin{proof}
  With a reference measure fixed as in \eqref{hh}, we have, along a
  path  $t \mapsto p_t$, $p_t = g_t p_0$, where $g_t =
  f_{p_t}/f_{p_0}$. Applying the definition of $\lambda_x$, we compute
  \begin{equation*}
    \lambda_x \Big(\,\frac{d}{dt} p_t\Big|_{t=0}\,\Big) = 
    \frac{d}{dt}\frac{f_{p_t}(x)}{f_{p_0}(x)}\Big|_{t=0} =
    \frac{d}{dt}\frac{e^{L_x(p_t)}}{e^{L_x(p_0)}}\Big|_{t=0} = dL_x
    \Big(\, \frac{d}{dt} p_t\Big|_{t=0}\,\Big).
  \end{equation*}
\end{proof}
\noindent%
In particular, local maxima of $L_x$ (points of so-called maximum
likelihood) do not depend on the reference measure.

\section{The metric and derivation of the bound}
With the observation-dependent one-forms in hand, we may now define
the Fisher-Rao Riemannian metric on ${\mathcal P} $. It is given by
\begin{equation*}
  {\mathbb I}(u,v)=\int_{x \in {\mathcal
      X}}\lambda_x(u)\lambda_x(v)\,dp\,, \text{for $u,v \in
    T_p {\mathcal P}$.}
\end{equation*}
Now that we have a metric, it is natural to consider $\nabla \theta $
instead of $d \theta $ in Proposition \ref{s2}. By the definition of
gradient, we have
\begin{equation*}
  |\nabla \theta(p)|^2 = d \theta (\nabla \theta (p)).
\end{equation*}
This equation and Proposition \ref{s2} now gives, for any
$v \in T_p {\mathcal P} $,
\begin{equation*}
  |\nabla \theta(p)|^2 = \int_{x \in {\mathcal X}}\hat \theta (x)
  \lambda_x(\nabla \theta (p))\,d{p}
  =\int (\hat\theta(x) - \theta(p)) \lambda_x (\nabla \theta
  (p))\,d{p}.
\end{equation*}
The second equality holds because
$\int_{x \in {\mathcal X}} \lambda_x(\nabla \theta(p))\,d{p}=0$, by
Lemma \ref{s2}. Applying the Cauchy-Schwartz inequality to the
right-hand side gives
\begin{align*}
  |\nabla \theta(p)|^2 &\le \left(\int_{x \in {\mathcal
                         X}}(\hat\theta(x) - \theta(p))^2\,d{p}\right)^{1/2}\left(\int_{x \in
                         {\mathcal X}}\lambda_x(\nabla \theta (p))\lambda_x(\nabla \theta (p))\,d{p}\right)^{1/2}\\
                       &= \sqrt{{\mathbb V}(\hat\theta\,|\,p))}\, \sqrt{{\mathbb I}(\nabla \theta (p),\nabla
                         \theta (p))} = \sqrt{{\mathbb V}(\hat\theta\,|\,p))}\enspace|\nabla \theta(p)|.
\end{align*}
The Cram\'er-Rao bound now follows.

\section{The bound in terms of Fisher information}
Theorem \ref{_1} is coordinate-free formulation. To recover the more
usual statement of the Cram\'er-Rao bound, let $\phi_1,\ldots,\phi_k$
be local coordinates on ${\mathcal P} $, the space of models on
${\mathcal X} $, and
$\frac{\partial }{\partial \phi_1}, \ldots, \frac{\partial }{\partial
  \phi_k}$ the corresponding vector fields on ${\mathcal P} $,
characterised by
\begin{equation*}
  d \phi_i \left(\frac{\partial }{ \partial \phi_j}\right)=\delta_i^j.
\end{equation*}
Here $\delta_i^j=1$ if $i=j$ and is zero otherwise. Applying Lemma
\ref{s3}, the coordinate representation $I_{ij}$ of the Fisher-Rao
metric ${\mathbb I}$ is given by
\begin{align*}
  I_{ij}(p)={\mathbb I}\Big(\, \frac{\partial }{\partial \phi_i}(p), 
  \frac{\partial }{\partial \phi_j}(p)\,\Big)
  &= \int_{x \in {\mathcal X}}d \lambda_x\left(\frac{\partial
    }{\partial \phi_i}(p)\right) d \lambda_x\left(\frac{\partial
    }{\partial \phi_j}(p)\right)\, d{p}\\ 
  &= \int_{x \in {\mathcal X}}\left(\frac{\partial L_x
    }{\partial \phi_i}(p)\right)\left(\frac{\partial L_x
    }{\partial \phi_j}(p)\right)\, d{p}, 
\end{align*}
where $L_x(p)=\log f_p(x)$ is the log-likelihood. In statistics
$I_{ij}$ is known as the {\df Fisher information matrix}.

For the moment we continue to let $\theta $ denote an arbitrary
function on ${\mathcal P} $, and $\hat \theta $ an unbiased
estimate. Now $\nabla \theta $ is the gradient of $\theta $, with
respect to the metric ${\mathbb I} $. Since the coordinate representation of the metric is
$I_{ij}$, a standard computation gives the local coordinate formula
\begin{equation*}
  \nabla \theta  =  \sum_{i,j}I^{ij}\frac{\partial \theta }{\partial
    \phi_i}\frac{\partial }{\partial \phi_j},
\end{equation*}
where $\{I^{ij}\}$ is the inverse of $\{I_{ij}\}$. Regarding the lower
bound in Theorem \ref{_1}, we compute
\begin{align*}
  |\nabla \theta(p)|^2 &= {\mathbb I}(\nabla \theta(p), \nabla \theta
                         (p))
                         =\sum_{i,j,m,n}{\mathbb I}\Big(\,I^{ij}(p) \frac{\partial \theta }{\partial
                         \phi_i}(p)\frac{\partial }{\partial \phi_j}(p)\,,\, I^{mn}(p)
                         \frac{\partial \theta }{\partial
                         \phi_m}(p)\frac{\partial }{\partial \phi_n}(p)\,\Big)\\
                       &=\sum_{i,j,m,n}I^{ij}(p)I^{mn}(p)\frac{\partial
                         \theta }{\partial \phi_i}(p) \frac{\partial \theta }{\partial
                         \phi_m}(p)\,\,{\mathbb I}\Big(\, \frac{\partial }{\partial \phi_j}(p), 
                         \frac{\partial }{\partial \phi_n}(p)\,\Big)\\
                       &=\sum_{i,j,m,n}I^{ij}(p) I_{jn}(p) I^{mn}(p) \frac{\partial
                         \theta }{\partial \phi_i}(p) \frac{\partial \theta }{\partial
                         \phi_m}(p)\\
                       &=\sum_{i,m,n}\delta_i^n I^{mn}(p) \frac{\partial
                         \theta }{\partial \phi_i}(p) \frac{\partial \theta }{\partial
                         \phi_m}(p)= \sum_{i,m} I^{mi}(p)\frac{\partial
                         \theta }{\partial \phi_i}(p) \frac{\partial \theta }{\partial
                         \phi_m}(p).
\end{align*}
Theorem \ref{_1} now reads 
\begin{equation*}
  {\mathbb V}(\hat \theta \,|\, p) \ge \sum_{i,m} I^{mi}(p)\frac{\partial
    \theta }{\partial \phi_i}(p) \frac{\partial \theta }{\partial
    \phi_m}(p).
\end{equation*}

In particular, if we suppose $\theta$ is one of the coordinate
functions, say $\theta  = \phi_j$, then we obtain
\begin{equation*}
  {\mathbb V}(\hat \phi_j \,|\, p) \ge \sum_{i,m} I^{mi}(p)\frac{\partial
    \phi_j }{\partial \phi_i}(p) \frac{\partial \phi_j }{\partial
    \phi_m}(p)=\sum_{i,m}I^{mi}\delta_j^i \delta_j^m = I^{jj}(p),
\end{equation*}
the version of the Cram\'er-Rao bound to be found in statistics
textbooks.

\end{document}